\documentclass[a4paper]{amsart}
\usepackage{amsthm}
\usepackage{amsmath}
\usepackage{geometry}
\usepackage{amsfonts}
\usepackage{graphicx}
\usepackage{array}
\usepackage{amssymb}

\def\N{{\mathcal{N}}}
\newtheorem{theorem}{Theorem}
\newtheorem{corollary}[theorem]{Corollary}
\newtheorem{proposition}[theorem]{Proposition}
\newtheorem{definition}[theorem]{Definition}
\newtheorem{lemma}[theorem]{Lemma}

\numberwithin{equation}{section}
\numberwithin{theorem}{section}

\begin{document}
\title[Causality and skies]{\textsc{Causality and skies: is non-refocussing necessary?}}
\author{A. Bautista, A. Ibort}
\address{ICMAT and Depto. de Matem\'aticas, Univ. Carlos III de Madrid, Avda. de la
Universidad 30, 28911 Legan\'es, Madrid, Spain.}
\email{abautist@math.uc3m.es, albertoi@math.uc3m.es}
\author{J. Lafuente}
\address{Depto. de Geometr\'{\i}a y Topolog\'{\i}a, Univ. Complutense de
Madrid, Avda. Complutense s/n, 28040 Madrid, Spain.}
\email{jlafuente@mat.ucm.es}
\date{}
\thanks{This work has been partially supported by the Spanish MICIN grant
MTM 2010-21186-C02-02 and QUITEMAD P2009 ESP-1594.}

\begin{abstract}   It is shown that if $M$ is a strongly causal free of naked singularities space-time, then its causal structure is completely characterized by a partial order in the space of skies defined by means of a class non-negative Legendrian isotopies. It is also proved that such partial order is determined by the class of future causal celestial curves, that is, curves in the space of light rays which are tangent to skies and such that they determine non-negative sky-Legendrian isotopies.

It will also be proved that the space of skies $\Sigma$ equipped with Low's (or reconstructive) topology is homeomorphic and diffeomorphic to $M$ under the only additional assumption that $M$ separates skies, that is, that different points determine different skies.   The sky-separating property of $M$ being weaker than the ``non-refocussing'' property encountered in the previous literature is sharp and the previous result provides the answer to the question of what is the class of space-times whose causal structure, topology and differentiable structure can be reconstructed from their spaces of light rays and skies.  

Finally, the previous results allow a formulation of Malament-Hawking theorem in terms of the partial order defined on the space of skies.
\end{abstract}

\maketitle
\tableofcontents

\section{Introduction}
In a recent paper by Bautista \emph{et al} \cite{Ba14} it was shown, following a suggestion by Low \cite{Lo01, Lo06},  that if $M$ is a strongly causal, free of naked singularities, non-refocussing space-time $M$, then the sky map is a diffeomorphism, that is, the topology and differentiable structure of $M$ can be recovered from the natural topology and differentiable structure on its space of skies $\Sigma$ and light rays $\mathcal{N}$.  Moreover if $M_1$ and $M_2$ are two strongly causal space-times (where $M_2$ is null non-conjugate), and $\phi \colon \mathcal{N}_1
 \to \mathcal{N}_2$ is a diffeomorphism mapping causal celestial curves into causal celestial curves then there exists a conformal immersion  $\Phi \colon M_1 \to M_2$ such that $\phi (\gamma) = \Phi \circ \gamma$ for any light ray $\gamma\in \mathcal{N}_1$ and conversely (\cite[Thm. 4]{Ba14}).    Let us recall that a celestial curve $\Gamma$ is a differentiable curve  in the space of light rays which is tangent everywhere to a sky and it is called past (future) causal if it defines a non-negative (non-positive) Legendrian isotopy of skies.  The class of causal celestial curves emerges thus as the relevant geometrical structure on $\mathcal{N}$ characterizing the original conformal class of the Lorentzian metric on $M$.  Moreover  the previous conditions cannot be weakened as the examples discussed in \cite{Ba14} show.

Let us recall that strongly causal space-times are natural candidates for a reconstruction theorem because they constitute a class of space-times whose spaces of unparametrized null geodesics are smooth manifolds (see for instance \cite[Prop. 2.1. and ff.]{Lo89}).     We will assume in what follows that the space-time $M$ is strongly causal, hence time-orientable, and free of naked singularities (the later condition guaranteeing that its space of light rays is a Hausdorff space).     It is important to recall that the space of light rays carries a canonical contact structure that provides an additional piece of geometry relevant in the analysis that follows.

Each event $x \in M$ determines the congruence of light rays $S(x)$ passing through it and called the sky of $x$.   There is a natural map from $M$ into the set of skies $\Sigma$ called the sky map $S\colon M \to \Sigma$, $x \mapsto S(x)$.    The reconstruction theorem for the topological (respect., differentiable) structure of $M$ will consists in determining under what conditions the sky map $S$ is a homeomorphism (respect., a diffeomorphism), that is, under what conditions the space of skies $\Sigma$ inherits a natural topology  (respect., a smooth manifold structure) from $\mathcal{N}$ such that the sky map is a homeomorphism (a diffeomorphism).  

%%%%%%%%%%%%%%
 
Simple examples, like the Einstein cylinder $\mathbb{R} \times \mathbb{S}^{m-1}$ equipped with the standard product metric $\mathbf{g}=-dt^2\oplus \mathbf{h}$, where $\mathbf{h}$ is the induced Euclidean metric on $\mathbb{S}^{m-1}$, show that even globally hyperbolic spaces could have many-to-one sky maps.   Thus a natural condition that has to be imposed on $M$ is that the sky map is injective or, in other words, that skies separate events, i.e., given two different events $x\neq y$ their skies are different $S(x) \neq S(y)$.  Thus for a strongly causal space-time with the property that skies separate events, the sky map $S$ is invertible with inverse $P$ (called the `parachute map').  Now equipping $\Sigma$ with the natural topology induced from $\mathcal{N}$ it is easy to show that $S$ is continuous.   

In order to guarantee that the sky map is open Low introduced in \cite{Lo93}, \cite{Lo01} and \cite{Lo06} an apparently weaker property called \emph{non-refocusing}: a space-time $M$ is refocusing at $x \in M$ if there exists an open neighbourhood $U$ of $x$ such that for every open neighbourhood $V\subset U$ of $x$ there exists $y\notin U$ such that every null geodesic through $y$ enters $V$.  This property has been studied in depth in \cite{K} and plays an important role in the proofs given in \cite[Prop. 3]{Ba14} and \cite[Prop. 4.1]{K}, that the sky map $S \colon M\rightarrow \Sigma$ is a homeomorphism.  We claim that the hypothesis of non-refocusing is not necessary if the space-time $M$ is sky-separating, i.e., any space-time with skies separating events is homeomorphic and diffeomorphic to its space of skies.  This conclusion will be reached as a consequence of Thm. \ref{theo-regular} at the end of Section \ref{non-refocussing}.

Moreover, and this will constitute the first objective of this project, it is expected that not only the topological structure, but the causal structure of $M$ could be characterized in terms of the topological structure of $\N$ and $\Sigma$ too.   Again, it was conjectured by R. Low that two events in a space-time are causally related iff their corresponding skies, which are Legendrian knots with respect to the canonical contact structure in the space of null geodesics, are
linked.  Recently it was shown by Chernov and Rudyak \cite{Ch08} and Chernov and Nemirovski \cite{Ch10} that Low's conjecture is actually true in a globally hyperbolic space with a Cauchy surface whose
universal covering is diffeomorphic to an open domain in $\mathbb{R}^{n}$.
A fundamental role in such analysis is played by the partial relation defined by non-negative Legendrian isotopies. Recently Chernov and Nemirovski \cite{Ch14} had extended the previous ideas to show that the causal structure of a simply connected globally hyperbolic space-time $M$ can be reconstructed from the partial ordering in the universal covering of Legendrian isotopy class of the fibres of the sphere bundle of a smooth Cauchy surface.

In Sect. \ref{partial_order_skies} it will be shown that the causal structure of $M$ can be recovered from a partial ordering introduced in the space of skies by a restricted class of non-negative Legendrian isotopies called sky isotopies.
Without entering in the analysis of Low's Legendrian conjecture here, it will be shown that the analysis of the causal structure of $M$ in terms of $\Sigma$ is deeply related to the study of celestial curves.   It will be shown that celestial curves are in correspondence with a class of null curves that will be called twisted null curves.   The
causal structure of the original space-time will be characterized completely at the end of Sect. \ref{causal} in terms of the partial order relation induced in the space of skies by future (past) causal twisted null  curves.

Finally, the proof that the space of skies of a sky-separating space-time is homeomorphic (and diffeomorphic)
to the original space proceeds by constructing a basis for the reconstructive (or Low's) topology by means of 
regular open subsets of $\Sigma$, where `regular' here means that the corresponding
tangent spaces to the skies elements of the open set `pile up' nicely
defining a regular submanifold in the tangent space to $\mathcal{N}$.   The definition and discussion of the main properties of regular sets constitutes the core of Section \ref{non-refocussing}, where again the properties of twisted causal null curves
will be used in a critical way.

Thus we offer an answer to the question of characterizing a large class of space-times $M$ such that the pair $(\mathcal{N}, \Sigma)$ is capable of reconstructing the causal, topological and differentiable structures of $M$.   However the question of what is the largest class of space-times such that two different skies which are related by a future causal celestial curve are topologically linked as stated in Low's conjecture is still open.

%\newpage

\section{The space of light rays and the space of skies}

\subsection{The space of light rays}
Let $M$ be a second countable paracompact $m$-dimensional smooth manifold and $\mathcal{C}$ a conformal class of Lorentzian metrics of signature $(-+\cdots +)$ such that $M$ becomes a time-orientable strongly causal space-time.   We will denote by $\mathbf{g}$ a representative metric on $\mathcal{C}$ and a time-like vector field $T$ determining a time-orientation on $M$ will be fixed in what follows.     

Let $\mathcal{N}$ denote the space of unparametrized inextensible future-oriented null geodesics, called in what follows light rays, i.e., $\mathcal{N}$ is the space of equivalence classes of inextensible smooth null curves $\gamma \colon I \to M$, with $I$ an interval in $\mathbb{R}$, such that $\nabla_{\gamma'} \gamma' = 0$, $g(\gamma', T) < 0$, and two such curves are equivalent if they are related by an affine reparametrization for the chosen representative $\mathbf{g}$ of the conformal class $\mathcal{C}$.     

We will consider in what follows the fibre bundle $\mathbb{N}$ over $M$ consisting of nonzero null vectors, and the corresponding components of future (past) null vectors $\mathbb{N}^\pm$.  If we denote $\mathbb{N}_x^+ = \{ v \in \mathbb{N}_x \mid v \neq 0, g_x(v, T(x)) < 0 \}$ and   $\mathbb{N}_x^- = \{ v \in \mathbb{N}_x \mid v \neq 0, g_x(v, T(x)) > 0 \}$, we have $\mathbb{N}^\pm = \bigcup_{x \in M} \mathbb{N}^\pm_x$ and $\mathbb{N} = \mathbb{N}^+ \cup \mathbb{N}^-$.  We will denote by $\pi \colon \mathbb{N} \to M$ the restriction of the canonical tangent bundle projection $TM \to M$ to $\mathbb{N}$ (and $\mathbb{N}^\pm$).     

We will denote again the canonical projection $\pi \colon \mathbb{PN}^+ \to M$, where $\mathbb{PN}^+$ denotes the quotient space of $\mathbb{N}^+$ by the action of the multiplicative group of positive real numbers $\mathbb{R}^+$ by scalar multiplication.    Notice that there is a canonical surjection $\sigma \colon \mathbb{PN}^+ \to \mathcal{N}$, given by $\sigma ([u]) = \gamma_{[u]}$, where $\gamma_{[u]}$ (or $[\gamma_u]$ as it will be used in what follows too) denotes the unparametrized geodesic containing $\gamma_u$ and $\gamma_u (t)$ indicates the unique future parametrized geodesic such that $\gamma_u (0) = \pi (u)$, and $\gamma'_u (0) = u$.   Moreover, because $\gamma_{\lambda u}(t) = \gamma_u (\lambda t)$, $u \in \mathbb{N}^+$ and the previous notation is consistent.

\subsection{The smooth structure of $\mathcal{N}$}\label{atlas}
The space of light rays $\mathcal{N}$ can be equipped with the structure of a second countable paracompact smooth manifold of dimension $2m-3$, if $\dim M = m$, and such that the map $\sigma$ becomes a submersion, in two different ways.   We will succinctly describe them in the following paragraphs.  

First,  we can use the local structure of $M$, i.e., because $M$ is strongly causal, given any event $x\in M$, there exists a globally hyperbolic neighbourhood $U_x$ of $x$ and a local smooth Cauchy hypersurface $C_x \in U_x$ \cite{Mi08}.  We can take $U_x$ small enough such that it is contained in a local chart of $M$.  Hence we can define an atlas for $\mathcal{N}$ as follows, select for any event $x\in M$ a globally hyperbolic open neighbourhood $U_x$ as before with Cauchy hypersurface $C_x $.   We consider the restriction of the projective bundle $\mathbb{PN}^+$ to $C_x$ and we denote it by $\mathbb{PN}^+(C_x)$.   There is a natural embedding $i_x \colon \mathbb{PN}^+(C_x) \to \mathbb{PN}^+$.  Then the composition $\sigma \circ i_x \colon \mathbb{PN}^+(C_x) \to \mathcal{N}$ will provide the charts of the atlas we are looking for and the open sets $\mathcal{U}_x = \sigma \circ i_x ( \mathbb{PN}^+(C_x)) \subset \mathcal{N}$ will be the domains of the corresponding charts (see \cite[Sect. 2.3]{Ba14} for more details).

Alternatively, we can induce a smooth structure on $\mathcal{N}$ from the smooth structure of the bundle $\mathbb{N}^+$ by considering the foliation defined by the leaves of the integrable distribution generated by the vector fields $X_\mathbf{g}$ and $\Delta$, where $X_\mathbf{g}$ denotes the geodesic spray of a fixed representative metric in the conformal class $\mathcal{C}$ and $\Delta$ is the dilation or Euler field.    Because $[X_\mathbf{g}, \Delta] = X_\mathbf{g}$, the distribution $D = \mathrm{span} \{ \Delta, X_g\}$ is integrable and denoting by $\mathcal{D}$ the corresponding foliation, we have that the space of leaves $\mathbb{N}^+ / \mathcal{D} \cong \mathcal{N}$.    If $M$ is strongly causal it can be shown that $\mathcal{D}$ is a regular foliation and the space of leaves inherits a smooth structure from $\mathbb{N}^+$.   Again, it is not hard to show that both smooth structures coincide.

\subsection{The tangent bundle $T\mathcal{N}$ and the contact structure of $\mathcal{N}$}

Let $\Gamma \colon (-\epsilon, \epsilon) \to \mathcal{N}$ be a differentiable curve such that $\Gamma (0) = \gamma$ and let $\chi (s,t) \colon (-\epsilon, \epsilon) \times I \to M$ be a geodesic variation by null geodesics of a parametrization $\gamma (t)$ of the null geodesic $\gamma$, that is, $\chi$ is a smooth function such that $\chi (s,t) = \gamma_s (t)$ are null geodesics, $\gamma_0 (t)$ is a parametrization of $\gamma$, and $[\gamma_s] = \Gamma (s)$ where $[\gamma_s]$ denotes the unparametrized geodesic containing $\gamma_s$.  Then the vector field along $\gamma$ defined by $J = \partial \chi /\partial s \mid_{s = 0}$ is a Jacobi field.  The set of Jacobi fields along $\gamma (t)$ will be denoted by $\mathcal{J} (\gamma )$  and they satisfy the second order differential equation:
$$
J'' = R(\gamma', J) \gamma' \, ,
$$
where $J'$ denotes the covariant derivative of $J$ along $\gamma'(t)$.  Notice that since the geodesic variation $\chi$ is by null geodesics, we have $\langle J, \gamma' \rangle =$ constant and we denote by $\mathcal{L}(\gamma)$ the linear space of Jacobi fields satisfying this property.

Equivalence classes of curves $\Gamma (s)$ possessing a first order contact define tangent vectors to $\mathcal{N}$ at $\gamma$, hence tangent vectors at $\gamma$ correspond to equivalence classes of Jacobi fields with respect to the equivalence relation defined by reparametrization of the geodesic variation $\chi$.   Such reparametrizations will correspond to Jacobi fields of the form $(at + b) \gamma'(t)$, then there is a canonical projection $\mathcal{L}(\gamma) \to T_\gamma \mathcal{N}$, mapping each Jacobi field $J$ into a tangent vector $[J] =J \,\, \mathrm{mod} \gamma'$ whose kernel is given by Jacobi fields proportional to $\gamma'$.  In what follows the tangent vectors $[J]$ will be denoted again as $J$ unless there is risk of confusion.  

There is a canonical contact structure on $\mathcal{N}$ defined by the maximally non-integrable hyperplane distribution 
$\mathcal{H}_\gamma \subset T_\gamma \mathcal{N}$ formed by the vectors orthogonal to their supporting light ray, i.e., 
\begin{equation}\label{contact}
\mathcal{H}_\gamma = \{ J \in T_\gamma \mathcal{N} \mid \langle J, \gamma' \rangle = 0  \} \, .
\end{equation}
It is easy to show that $\mathcal{H}$ does not depend on the representative metric used to define, the representative $J$ chosen for the tangent vector, or the parametrization $\gamma(t)$ we chose for the light ray $\gamma$.    

Let us recall that if $X$ is a contact manifold with contact distribution a maximally non-integrable codimension one distribution $\mathcal{H}$, the contact structure is said to be exact or co-orientable if there exists a globally defined 1-form $\alpha$,  such that $\mathcal{H} = \ker \alpha$ and such 1-form is called a contact 1-form for the contact structure $\mathcal{H}$.  

It is obvious that the canonical contact structure $\mathcal{H}$ on $\mathcal{N}$, Eq. (\ref{contact}), can be locally defined by the family of 1-forms $\alpha^{x}$ defined on the open sets $\mathcal{U}_x$ of the atlas described in Sect. \ref{atlas} above and given by the explicit formula:
$$
\alpha_\gamma^{x} \colon J \mapsto \langle J, \gamma' \rangle \, ,
$$ 
where the parametrization $\gamma (t)$ of the light ray $\gamma$ is determined by the Cauchy surface $C_x \subset U_x$ and $\langle T(x), \gamma'(0) \rangle = -1$.  The local 1-forms $\alpha^{x}$ do not define a global 1-form, however because $\mathcal{N}$ is paracompact we can use a partition of the unity subordinated to a locally finite refinement of the open covering $\{ \mathcal{U}_x \}$ of $\mathcal{N}$ defined by family of globally hyperbolic open neighbourhoods $\{U_x \mid x \in M\}$, and paste the local 1-forms to define a globally defined 1-form whose kernel is $\mathcal{H}$.  

Notice however the space of not oriented unparametrized null geodesics still carries a canonical contact structure (defined by the same formula above, Eq. (\ref{contact})s ) which is not co-oriented.

%%%%%%%%%%%%%%%%%%%%%%%%%%%%%%
%%%%%%%%%%%%%%%%%%%%%%%%%%%%%%

%\newpage

\section{Reconstruction of the causal structure}\label{causal}

\subsection{The space of skies and its topology}

As it was explained in the introduction, the sky of an event is the congruence of light rays passing through it. Thus if $x\in M$ denotes an event, the corresponding sky will be denoted either by $S(x)$ or $X$.  Then $S(x) = \{ \gamma \in \mathcal{N} \mid x \in \gamma \}$.  Notice that there is a canonical map $\sigma_x \colon \mathbb{PN}^+_x \to S(x)$, $\sigma_x ([u]) = \gamma_{[u]}$.   Clearly the sky $S(x)$ as a submanifold of $\mathcal{N}$ is diffeomorphic to the sphere of dimension $m-2$.  The family of all skies will be denoted by $\Sigma$, that is,
$$
\Sigma = \{ X = S(x) \mid x \in M \} \, ,
$$
and the canonical map $S \colon M\to \Sigma$, $x \mapsto S(x)$, is called the sky map.   The sky map is clearly surjective, however it doesn't have to be injective as indicated in the introduction.  Hence we will say that $M$ separates skies if $S$ is injective, that is, if $x \neq y$, then $S(x) \neq S(y)$.  If $M$ separates skies, the map $P \colon \Sigma \to M$, inverse to the sky map, is well defined and will be called the parachute map.

The space of skies $\Sigma$ carries a canonical topology called the reconstructive topology defined as follows.  Let $	\mathcal{U} \subset \mathcal{N}$ be an open set, then consider the set of all skies $X$ such that $X \subset \mathcal{U}$.  We will denote this set by $\Sigma (\mathcal{U})$.  It is clear that the family of sets $\Sigma (\mathcal{U})$ satisfies $\Sigma (\mathcal{U}) \cap \Sigma (\mathcal{V}) = \Sigma (\mathcal{U} \cap \mathcal{V})$, then they constitute a basis for a topology on $\Sigma$ called the reconstructive topology. 

It is easy to prove that the sky map $S$ is continuous with respect to the reconstructive topology.  However it is not obvious if it is open or not.  As it was discussed in the introduction it is one of the objectives of this paper to determine under what conditions $S$ is open, i.e., $P$ continuous, or not.

We will end these remarks by observing that if $X = S(x)$ is a sky, then given $\gamma \in X$, a tangent vector $J$ to $X$ at $\gamma$ is determined by a geodesic variation such that all their geodesics pass through the point $x$ at time 0, then $J(0) = 0$.  This implies that $\langle J, \gamma' \rangle = 0$ for all $J \in T_\gamma X$ and $TX \subset \mathcal{H}$.  Thus skies are Legendrian spheres because, in addition, $2m- 4 = \dim \mathcal{H}_\gamma = 2 \dim T_\gamma X = 2(m-2)$.

\subsection{The partial order in the space of skies}

The canonical contact structure on $\mathcal{N}$ allows to define a natural partial ordering in the space of skies.  

Let us recall first that if $X$ is a co-oriented contact manifold with contact distribution $\mathcal{H} = \ker \alpha$ where $\alpha$ is a contact 1-form, a differentiable family $\Lambda_s$, $s \in [0,1]$, of diffeomorphic Legendrian submanifolds is called a Legendrian isotopy.  It is always possible to describe a Legendrian isotopy via a parametrization $F \colon \Lambda_0 \times [0,1] \to X$ verifying $F(\Lambda_0, s) = \Lambda_s \subset X$.   The map $F_s \colon \Lambda_0 \to \Lambda_s$, given by $F_s(\lambda) = F(\lambda, s)$ is a diffeomorphism for all $s \in [0,1]$.  
 The Legendrian isotopy $\Lambda_s$ is said to be non-negative (non-positive) if $(F^*\alpha)(\partial / \partial s) \geq 0$ (respect. $(F^*\alpha)(\partial / \partial s) \leq 0$) with $F$ a parametrization of $\Lambda_s$.   It is easy to check that the previous definition does not depend on the chosen parametrization.
 
 If we consider now the class $\mathcal{S}$ of Legendrian spheres on the contact manifold $X$, we can define a partial order on $\mathcal{S}$ by saying that $S_0 \prec S_1$, $S_0,S_1 \in \mathcal{S}$, if there exists a non-negative Legendrian isotopy $F \colon S_0 \times [0,1] \to X$, joining $S_0$ and $S_1$, i.e., such that $F_0(S_0 ) = S_0$, $F_1(S_0 ) = S_1$.  
 
 We will consider the previous ideas in the contact manifold $\mathcal{N}$ of light rays of a given space-time $M$.   The class $\mathcal{S}$ of Legendrian spheres in $M$ contains the space of skies $\Sigma$.   Then the partial order $\prec$ described before induces a partial order in $\Sigma$.  However we would like to restrict the previous partial order because it could happen that two skies $X_0 = S(x_0)$ and $X_1 = S(x_1)$ would be related, $X_0 \prec X_1$, but the non-negative Legendrian isotopy $X_s$ joining $X_0$ and $X_1$ will fall out of $\Sigma$, that is, not all Legendrian spheres $X_s$ will be the sky of a point $x_s \in M$.  
 
 Thus we will weaken the partial order $\prec$ by restricting the class of Legendrian isotopies to those consisting of skies.
 Hence let $F\colon X_0 \times [0,1] \to \mathcal{N}$ be a Legendrian isotopy such that $X_s = F_s(X_0)$ is the sky of $x_s\in M$, i.e., $X_s = S(x_s)$ and it defines a differentiable curve $\mu \colon [0,1] \to M$, given by $\mu(s) = x_s$.  Conversely,
 let $x_0 \in M$ be an event and $X_0 = S(x_0)$ its sky which is a Legendrian sphere,
then any differentiable curve $\mu \colon [0,1] \to M$ with $\mu (0) = x_0$ defines a Legendrian isotopy parametrized by the function $F^\mu \colon X_0 \times [0,1] \to \mathcal{N}$ given by $F^\mu (\gamma_{[u]}, s) = \gamma_{[u_s]}$, and $u_s \in \mathbb{N}_{\mu(s)}^+$ is the parallel transport of $u \in \mathbb{N}_{x_0}^+$ along $\mu$.  Notice that then $F^\mu$ is a Legendrian isotopy of skies and $F_s(X_0) = S(\mu (s))$, $s \in [0,1]$.  

We will call the Legendrian isotopies consisting of skies, sky isotopies and the corresponding partial order in the space of skies will be denoted by $\prec_\Sigma$. 

On the other hand there is a natural partial order relation in $M$ defined by the conformal class of the Lorentzian metric.  Thus given two events $x,y\in M$, we say that $y$ is in the causal future of $x$ and it will be denoted by $x \prec y$, if  $y\in J^{+}\left(x\right)$, i.e., $y$ can be reached by a future oriented causal curve starting at $x$.

Now it is simple to show that the curve $\mu \colon [0,1] \to M$ is causal past (future) iff $F^\mu$ is a non-negative (respect. non-positive) sky isotopy.   Hence we have the following characterization of causality in terms of definite sky isotopies (\cite[Prop. 4]{Ba14}).
 
\begin{proposition}
$x \prec y$ iff $X \prec_\Sigma  Y$.
 \end{proposition}
 
 The previous observations and results lead naturally to the following:
 
 \begin{definition}
 A continuous curve $\chiÊ\colon [0,1] \to \Sigma$ will be causal past (future) if it defines a non-negative (respect. non-positive) Legendrian isotopy in $\mathcal{N}$.   Two skies $X,Y \in \Sigma$ are said to be past (future) causally related if there is a causal past (future) curve $\chi$ such that $\chi (0) = X$ and $\chi(1)= Y$, and it will be denoted by $X \prec_{c} Y$ ($Y \prec_{c} X$).
 \end{definition}
 
As a consequence of the ``Twisted Curve Theorem'', Thm. \ref{mu-teorema}, the ``$\mu$-Lemma'', Lemma \ref{mu-lemma}, and Cor. \ref{differ} below it follows that the space-time $M$ is diffeomorphic and order isomorphic to the space of skies $\Sigma$ equipped with the partial order $\prec_c$ and the natural differentiable structure induced from the space of light rays $\mathcal{N}$. 

\begin{corollary}\label{order_iso}  Let $M$ be a strongly causal free of naked singularities and sky-separating space-time, then $M$ is diffeomorphic and order isomorphic to its space of skies $\Sigma$ 
\end{corollary}

\subsection{Celestial curves and twisted null curves}

As stated in the introduction, the reconstruction theorem in \cite{Ba14} asserts that the conformal structure of $M$ is captured by the class of causal celestial curves, that is by curves in $\mathcal{N}$ that are everywhere tangent to skies.
More formally:

\begin{definition}\label{celestial}
A non-zero tangent vector $J \in \widehat{T}_\gamma\mathcal{N}$, (with $\widehat{T}_\gamma\mathcal{N} = T_\gamma\mathcal{N} - \{\mathbf{0} \}$), will be called a celestial vector if there exists a sky $S\in \Sigma$ such that $J \in \widehat{T}_\gamma S$.    A differentiable curve $\Gamma \colon I \to \mathcal{N}$ is called a celestial curve if $\Gamma'(s)$ is a celestial vector for all $s \in I$.
\end{definition}
 
We will analyze in this section the relation existing between celestial curves and the causality properties of $M$ and of the space of skies $\Sigma$.   To do that we will introduce first the notion of twisted causal null curve that will prove to be useful in the arguments to follow.

\begin{definition}\label{definition1} 
A continuous curve $\mu:\left[a,b\right]\rightarrow M$ will be called a \emph{piecewise twisted null curve} if there exists a partition $a=s_0 < s_1 <\ldots < s_k = b$ such that for every $i=1,\ldots,k$:
\begin{enumerate}
\item[i.] $\left. \mu \right|_{\left(s_{i-1},s_i\right)}$ is differentiable.
\item[ii.] $\mathbf{g}\left(\mu'\left(s\right),\mu'\left(s\right)\right)=0$ for all $s\in \left(s_{i-1},s_i\right)$.
\item[iii.] $\mu'\left(s\right)$ and $\frac{D\mu'}{ds}\left(s\right)$ are linearly independent for all $s\in \left(s_{i-1},s_i\right)$.
\end{enumerate}

We say that $\mu$ is causal if $\mu\mid_{(s_{i-1},s_i)}$ is causal future (respect. causal past) for all $i = 1,\ldots,k$.
If $k=1$ then $\mu$ will be simply called \emph{twisted null curve}.
\end{definition}

Now it is clear that if we are given a parametrized null geodesic  $\gamma:\left[0,1\right]%
\rightarrow M$, a curve $\lambda :\left( -\epsilon ,\epsilon
\right)\rightarrow M$ verifying that $\lambda \left( 0\right)
=\gamma\left( 0\right)$, and $W\left( s\right) $ a null vector field along $\lambda $ such that $W\left( 0\right) =\gamma ^{\prime }\left( 0\right) $, the family of curves:
\begin{equation}
\mathbf{f}\left( s,t\right) =\mathrm{exp}_{\lambda \left( s\right)
}\left(tW\left( s\right) \right)
\end{equation}%
is a geodesic variation of $\gamma (t )$ formed by null geodesics with $\mathbf{%
f}\left( 0,t\right) =\gamma \left( t\right)$ and $J\left( t\right) =\frac{%
\partial \mathbf{f}}{\partial s}\left(0,t\right) $.   

If $\mu$ is a null curve then we may use $W (s) = \mu'(s)$ and obtain a geodesic variation of $\gamma$ that, in addition, defines a celestial curve in $\mathcal{N}$.  Actually more is true as it is shown by the following:

\begin{proposition}
\label{prop-existence-mu}\cite{Ba14} If the curve $\Gamma :\left[ 0,1\right]
\rightarrow \mathcal{N}$ with $\Gamma \left( s\right) =\gamma _{s}\in
\mathcal{N}$ is celestial then there exists a differentiable
null curve $\mu :\left[ 0,1\right] \rightarrow M$ such that $\gamma
_{s}\left( \tau \right) =\exp _{\mu \left( s\right) }\left( \tau \sigma
\left( s\right) \right) $ where $\sigma \left( s\right) \in \mathbb{N}^{+}_{\mu
\left( s\right) }$ is a differentiable curve proportional to $\mu ^{\prime
}\left( s\right) $ wherever $\mu$ is regular.
\end{proposition}

In fact, by construction, the curve $\mu$ in Prop. \ref{prop-existence-mu} runs the points in $M$ such that the celestial curve $\Gamma$ is tangent to their skies, in other words,  $\Gamma'\left(s\right)\in \widehat{T}S\left(\mu\left(s\right)\right)$ for all $s\in\left[0,1\right]$.   

As a consequence of the previous result, we have the following corollary.

\begin{corollary}\label{corol-mu}
Given a celestial curve $\Gamma:\left[ 0,1\right]\rightarrow \mathcal{N}$ such that $\Gamma'\left(s_0\right)\in \widehat{T}S\left(p_0\right)$, $0\leq s_0\leq 1$, then the curve $\mu \colon \left[ 0,1\right]\rightarrow M$ of the previous proposition \ref{prop-existence-mu} is unique verifying $\mu\left(s_0\right)=p_0\in M$. 
\end{corollary}

\begin{proof}
Consider that there exists $\mu_1,\mu_2: \left[0,1\right]\rightarrow M$ associated to $\Gamma$ in the sense of proposition \ref{prop-existence-mu} and verifying $\mu_1\left(s_0\right)=\mu_2\left(s_0\right)=p_0$ for $s_0\in \left[0,1\right]$.
Let us define the set $A=\left\{s\in \left[0,1\right]:\mu_1\left(s\right)=\mu_2\left(s\right)  \right\}$. 
Clearly, $A$ is not empty and closed in $\left[0,1\right]$. 
Consider a causally convex and normal neighbourhood $U\subset M$ of $p_0$. 
Since $U$ is open, then there exist $\delta >0$ such that $\mu_i\left(\left(s_0-\delta , s_0 +\delta\right)\right)\subset U$ for $i=1,2$ (eventually if $s_0=0$ then we consider $\mu_i\left(\left[0 , \delta\right)\right)\subset U$ and analogously for $s_0=1$). 
Let us suppose that for $s\in \left(s_0-\delta, s_0+\delta\right)$ we have that $\mu_1\left(s\right)\neq\mu_2\left(s\right)$ and since $U$ is causally convex, then the segment of the light ray $\Gamma\left(s\right)=\gamma_s\in \mathcal{N}$ connecting $\mu_1\left(s\right)$ and $\mu_2\left(s\right)$ is totally contained in $U$ and, moreover since $\Gamma'\left(s\right)\in \widehat{T}S\left(\mu_1\left(s\right)\right)\cap \widehat{T}S\left(\mu_2\left(s\right)\right)$, then the points $\mu_1\left(s\right)$ and $\mu_2\left(s\right)$ are mutually conjugated along $\gamma_s$ but, in virtue of \cite[Prop. 10.10]{On83}, this is not possible in a normal neighbourhood contradicting $U$ is normal.
Then we have that $\mu_1\left(s\right)=\mu_2\left(s\right)$ and hence the set $A$ is also open in $\left[0,1\right]$.
Since $A$ is open, closed and not empty in $\left[0,1\right]$ then $A=\left[0,1\right]$ and we conclude that $\mu_1=\mu_2$. 
\end{proof}

Given a celestial curve $\Gamma$ the unique curve $\mu$ associated to it in the sense of Prop. \ref{prop-existence-mu}  passing by $p_0 \in S^{-1}(X_0)$ will be called the ``dust'' of $\Gamma$ by $X_0$ and denoted by $\mu_{X_0}^\Gamma$.   The previous arguments can be made more precise by proving that the dust of a celestial curve is a twisted null curve.  This is the content of the next Lemma.

\begin{lemma}[$\mu$-Lemma]\label{mu-lemma}
Let $\Gamma:\left[0,1\right]\rightarrow \mathcal{N}$ be a celestial curve such that $\Gamma'\left(0\right)\in \widehat{T}X_0$ with $X_0\in \Sigma$. 
Then there exists a unique curve $\chi_{X_0}^{\Gamma}\colon \left[0,1\right]\rightarrow \Sigma$ 
such that it is continuous in Low's topology and verifies  $\chi_{X_0}^{\Gamma}\left(0\right)=X_0$ and $\Gamma'\left(s\right)\in \widehat{T}\chi_{X_0}^{\Gamma}\left(s\right)$.
Moreover, the dust curve $\mu^{\Gamma}_{X_0}$ is a piecewise twisted null curve in $M$ running along the image of $S^{-1} \circ \chi_{X_0}^{\Gamma}$.

Conversely, given a regular twisted null curve $\mu\colon \left[0,1\right]\rightarrow M$ such that $\mu\left(0\right)=x_0=S^{-1}\left(X_0\right)$, $\mu'(0) \neq 0 \neq \mu'(1)$, then the curve $\Gamma^{\mu}:\left[0,1\right]\rightarrow \mathcal{N}$ defined by the variation of null geodesics $\mathbf{x}:\left[0,1\right]\times I \rightarrow M$ such that 
\[
\mathbf{x}\left(s,t\right)=\mathrm{exp}_{\mu\left(s\right)}\left(t\mu'\left(s\right)\right)=\left.\Gamma^{\mu}\left(s\right)\right|_{t}
\]
is celestial with $\Gamma'\left(0\right)\in \widehat{T}X_0$ and $\chi_{X_0}^{\Gamma}\left(s\right)=S\left(\mu\left(s\right)\right)$.
\end{lemma}

\begin{proof}
Let $\Gamma:\left[0,1\right]\rightarrow \mathcal{N}$ be a celestial curve such that $\Gamma\left(s\right)=\gamma_s\in \mathcal{N}$ and  $\Gamma'\left(0\right)\in \widehat{T}X_0$ with $X_0=S\left(x_0\right)\in \Sigma$. 
By corollary \ref{corol-mu}, there exists a unique differentiable curve $\mu: \left[0,1\right]\rightarrow M$ and a partition 
\[
\left\{0=a_1 \leq b_1 < a_2 \leq b_2 <\cdots < a_{n-1}\leq b_{n-1}<a_n\leq b_n =1 \right\}\subset \left[0,1\right]
\]
such that 
\begin{equation}\label{equat-mu-lemma-1}
\gamma_{s}\left( \tau \right) =\exp _{\mu \left( s\right) }\left( t \sigma\left( s\right) \right) 
\end{equation}
where $\sigma:\left[0,1\right]\rightarrow \mathbb{N}$ is a differentiable curve verifying $\sigma\left(s\right)=\lambda_k\left(s\right)\mu'\left( s\right)$ for $s\in \left(b_k,a_{k+1}\right)$ and $\lambda_k$ differentiable with $k=1,\ldots , n-1$. 
This curve $\mu$ also verifies $\mu\left(s\right)=p_k\in M$ for all $s\in \left[a_k,b_k\right]$.

Now, we can define the curve $\chi_{X_0}^{\Gamma}=S\circ \mu: \left[0,1\right]\rightarrow \Sigma$.
Recall that for an open set $\mathcal{U}\subset \mathcal{N}$ containing a sky $X\in \Sigma$, the set of all skies contained in $\mathcal{U}$ is denoted as $\Sigma\left(\mathcal{U}\right)$. 
By the definition of the Low's topology, the set $\Sigma\left(\mathcal{U}\right)$ is open in $\Sigma$ and these collection of open sets forms a basis at $X$.

In order to show that $\chi_{X_0}^{\Gamma}$ is continuous, we will show that, given any $\mathcal{U}\subset \mathcal{N}$ containing a sky $S\left(\mu\left(s\right)\right)\in \Sigma$ then $\left(\chi_{X_0}^{\Gamma}\right)^{-1}\left(\Sigma\left(\mathcal{U}\right)\right)$ is open in $\left[0,1\right]$ is verified. 
So, take any $s\in \left[0,1\right]$ and consider an open set $\mathcal{U}\subset \mathcal{N}$ such that $\chi_{X_0}^{\Gamma}\left(s\right)\subset \mathcal{U}$ and then $\chi_{X_0}^{\Gamma}\left(s\right)\in \Sigma\left(\mathcal{U}\right)$.
Choose a collection of nested intervals $I_n^s\subset \mathbb{R}$ such that $\{s \} = \bigcap_{n} I_n^s$. 
Let us suppose that there exists $s_n\in I_n^s$ such that $\chi_{X_0}^{\Gamma}\left(s_n\right)\notin \Sigma\left(\mathcal{U}\right)$. 
Then there is a light ray $\gamma_n\in \chi_{X_0}^{\Gamma}\left(s_n\right)\in \Sigma$ such that $\gamma_n\notin \mathcal{U}$.
Recall that a light ray is fully determined by a point $p\in M$ and a direction $[v] \in \mathbb{PN}^+_p$, so $\gamma_n$ can be defined by $\mu\left(s_n\right)\in \gamma_n \subset M$ and a null direction $[v_n] \in \mathbb{PN}^+_{\mu\left(s_n\right)}$.  Since $\lim \mu\left(s_n\right) =\mu\left(s\right)$ and due to the compactness of the fibres $\mathbb{PN}^+_{\mu\left(s_n\right)}$, then with no lack of generality taking a subsequence of $[v_n]$ if necessary, there exists a direction $[v]\in \mathbb{PN}^+_{\mu\left(s\right)}$ defining, together with $\mu\left(s\right)$, the light ray $\gamma$ such that $\lim \gamma_n = \gamma \in \chi_{X_0}^\Gamma(s) \subset \mathcal{U}$.

But since $\mathcal{U}$ is open, there exists an integer $K$ such that for every $n>K$ we have that $\gamma_n\in\mathcal{U}$ contradicting that $\chi_{X_0}^{\Gamma}\left(s_n\right)\notin \Sigma\left(\mathcal{U}\right)$.
Therefore there exist $I_n^s$ such that  $\chi_{X_0}^{\Gamma}\left(s_n\right)\in \Sigma\left(\mathcal{U}\right)$ and hence $\left(\chi_{X_0}^{\Gamma}\right)^{-1}\left(\Sigma\left(\mathcal{U}\right)\right)$ is open in $\left[0,1\right]$.
 
To obtain the dust $\mu^{\Gamma}_{X_0}$, we will cut off the segments $\left.\mu\right|_{\left(a_k,b_k\right)}$ from $\mu$ and glue together the segments $\left.\mu\right|_{\left[b_k,a_{k+1}\right]}$.   We call $c_1 = 0$ and for 
every $k=1,\ldots,n-1$, let us define $c_{k+1} = a_{k+1} - \sum_{i=1}^{k}\left(b_i -a_i\right)\in \left[0,1\right]$ and consider the change of parameter $h_k:\left[c_k,c_{k+1}\right]\rightarrow \left[b_k,a_{k+1}\right]$ defined by $h_k\left(\tau\right)=\tau +a_{k+1} -c_{k+1}$. 
Since $\mu$ is differentiable and $h_k$ is a diffeomorphism for every $k=1,\ldots,n-1$ then $\overline{\mu}_k\left(\tau\right)=\mu\circ h_{k}\left(\tau\right)$ is differentiable for $\tau\in \left(c_k,c_{k+1}\right)$.
Moreover, since $\overline{\mu}'_k\left(\tau\right)=\mu'\left(h_{k}\left(\tau\right)\right)$ then 
\[
\mathbf{g}\left(\overline{\mu}'_k\left(\tau\right),\overline{\mu}'_k\left(\tau\right)\right)=\mathbf{g}\left(\mu'_k\left(h_{k}\left(\tau\right)\right),\mu'_k\left(h_{k}\left(\tau\right)\right)\right)=0
\]
for $\tau\in \left(c_k,c_{k+1}\right)$.
Also, the covariant derivatives verify 
\[
\frac{D\overline{\mu}'_k\left(\tau\right)}{d \tau} = h''_k\left(\tau\right)\mu'\left(h_k\left(\tau\right)\right) + \left(h'_k\left(\tau\right)\right)^2 \frac{D\mu'\left(h_k\left(\tau\right)\right)}{ds} = \frac{D\mu'\left(h_k\left(\tau\right)\right)}{ds}
\]
then denoting $J_s$ as the Jacobi field along $\gamma_s$ defined by the variation \ref{equat-mu-lemma-1}, we have $J_s\left(0\right)=\mu'\left(s\right)$ and 
\[
J'_s\left(0\right)=\frac{D\sigma\left(s\right)}{ds}=\frac{D\left(\lambda_k\left(s\right)\mu'\left(s\right)\right)}{ds}=\lambda'_k\left(s\right)\mu'\left(s\right)+\lambda_k\left(s\right)\frac{D\mu'\left(s\right)}{ds}
\]
for $s\in \left(b_k,a_{k+1}\right)$.
Since $\Gamma$ is celestial, then $J_s\neq 0\left(\mathrm{mod}\gamma'_s\right)$ and so, $\frac{D\mu'\left(s\right)}{ds}$ is not proportional to $\mu'\left(s\right)$ for $s\in \left(b_k,a_{k+1}\right)$, therefore $\frac{D\overline{\mu}'_k\left(\tau\right)}{d \tau}$ and $\overline{\mu}'_k\left(\tau\right)$ are linearly independent for $\tau\in \left(c_k,c_{k+1}\right)$. 
We have shown that for any $k=1,\ldots,n-1$ the curves $\overline{\mu}_k$ are twisted null curves.
Since $h^{-1}_k\left(a_{k+1}\right)=h^{-1}_{k+1}\left(b_{k+1}\right)$ then all the segments $\overline{\mu}_k$ glue together continuously.
Therefore we can define, with no ambiguity, the curve $\mu_{X_0}^{\Gamma}:\left[0,a\right]\rightarrow M$ such that $\mu_{X_0}^{\Gamma}\left(\tau\right)=\overline{\mu}_k\left(\tau\right)$ if $\tau\in\left[c_k,c_{k+1}\right]$ for $k=1,\dots,n-1$ and $\left[0,a\right]=\cup_{k=1}^{n-1}\left[c_k,c_{k+1}\right]$. This curve $\mu_{X_0}^{\Gamma}$ is then a piecewise twisted null curve associated to the partition $\left\{0=c_1  < c_2  <\cdots < c_{n} = a \right\}\subset \left[0,a\right]$ and it is unique except by reparametrization. 

Conversely, let us consider a twisted null curve $\mu: \left[0,1\right]\rightarrow M$ such that $\mu\left(0\right)=x_0=S^{-1}\left(X_0\right)$. Then, we can define the variation of null geodesics $\mathbf{x}:\left[0,1\right]\times I \rightarrow M$ such that 
\[
\mathbf{x}\left(s,t\right)=\mathrm{exp}_{\mu\left(s\right)}\left(t\mu'\left(s\right)\right)=\gamma_s\left(t\right)
\]
which verifies $\gamma'_s\left(0\right)=\mu'\left(s\right)$.
Now, define the curve $\Gamma^{\mu}\left(s\right)=\gamma_s \in \mathcal{N}$ for every $s\in \left[0,1\right]$. 
The Jacobi field $J_s$ of the variation $\mathbf{x}$ along $\gamma_s$ verifies $J_s\left(0\right)=\mu'\left(s\right)=\gamma'_s\left(0\right)$ and $J'_s\left(0\right)=\frac{D\mu'}{ds}\left(s\right)$ and, since $\mu$ is twisted null then $\frac{D\mu'}{ds}$ is not proportional to $\gamma'_s$. 
Therefore $\left(\Gamma^{\mu}\right)'\left(s\right)=J_s\left(\mathrm{mod}\gamma'_s\right)\neq 0\left(\mathrm{mod}\gamma'_s\right)$ and hence \[
\left(\Gamma^{\mu}\right)'\left(s\right)\in \widehat{T}S\left(\gamma_s\left(0\right)\right)=\widehat{T}S\left(\mu\left(s\right)\right)
\]
then $\Gamma^{\mu}$ is celestial. 
\end{proof}

%%%%%%%%%%%%%%%%%%%%%%%%%%%%%%
%%%%%%%%%%%%%%%%%%%%%%%%%%%%%%

\subsection{Celestial curves and the partial order in the space of skies}\label{partial_order_skies}

We have already pointed it out that if $x \prec y$, then their corresponding skies are related $S(x) \prec_c S(y)$.   The discussion to follow will show that such relation can actually be refined by proving that in case of  $y \in I^+(x)$\footnote{Recall that $y\in I^+(x)$ means that there exists a future time-like  curve from $x$ to $y$.}, there exists a causal piecewise twisted null curve joining $x$ and $y$, hence relating the causal properties of $\Sigma$ to the existence of appropriate celestial curves.

\begin{theorem}[Twisted null curve theorem]\label{mu-teorema}
Let $p,q\in M$ such that $q\in I^+(p)$, then there exists a future piecewise twisted null curve $\mu$ joining $p$ to $q$.
\end{theorem}

To prove the previous Theorem we will need some lemmas.

\begin{lemma}\label{step1}   Let $M$ be a 3--dimensional space-time and $\gamma \colon I \to M$ be a future time-like geodesic.   Then there exists $\delta > 0$ such that for any $t\in (t_0,t_0+\delta]$, there exists a future twisted null curve $\mu$ joining $\gamma (t_0)$ to $\gamma (t)$.
\end{lemma}

\begin{proof}
Given the future time-like geodesic $\gamma \colon I \rightarrow M$ and $t_0\in I$, it is known, e. g. by \cite[\S 97]{L} and \cite[def. 7.13]{P}, that there exists a synchronous coordinate system $\left(U,\phi=\left(t,x,y\right)\right)$ with $\gamma\left(t_0\right)\subset U$ in which the metric $\mathbf{g}$ of $M$ can be written as 
\[
\left(g_{ij}\right)= 
\left( 
\begin{array}{ccc}
-1 & 0 & 0 \\
0 & g_{11} & g_{12} \\
0 & g_{12} & g_{22}
\end{array}
\right)
\]
where $g_{ij}\equiv g_{ij}\left(t,x,y\right)$ for $i,j=1,2$, $U$ is causally convex and the expression of the geodesic $\gamma$ in these coordinates is $\phi\left(\gamma\left(s\right)\right)=\left(s,0,0\right)\in \mathbb{R}^3$. 
For a point $\gamma\left( \overline{t} \right)\in U$, it is possible to find $R>0$ such that the compact set 
\[
U_0=\left\{\left(t,x,y\right):x^2+y^2\leq R, t_0\leq t \leq \overline{t} \right\}
\]
is contained in $U$.

As candidates for the required twisted null curve, we will study curves $\mu_r$ such that
\[
\phi\left(\mu_r\left(s\right)\right)=\left(f_r\left(s\right),r\left(1-\cos s\right),r\sin s\right)
\]   
where $0\leq r \leq R/2$ and $f_r=f_r\left(s\right)$ is a function.
If $\mu_r$ is a null curve, then $\mathbf{g}\left(\mu'_r,\mu'_r\right)=0$ and therefore 
\[
-\left(f'_r\left(s\right)\right)^2 + r^2 g_{11}\sin^2 s + 2r^2 g_{12}\sin s \cos s + r^2 g_{22}\cos^2 s = 0
\]
where $g_{ij}=g_{ij}\left(\phi\left(\mu_r\left(s\right)\right)\right)$.
Thus, we have a first order ordinary differential equation which describes a null curve passing through $\gamma\left(t_0\right)$
\begin{equation}\label{ode-mu}
\left\{
\begin{tabular}{l}
$f'_r\left(s\right)=r\sqrt{ g_{11}\sin^2 s + 2 g_{12}\sin s \cos s +  g_{22}\cos^2 s }$ \\
$f_r\left(0\right)=t_0$
\end{tabular}
\right.
\end{equation}
Since the metric in the hypersurfaces $\left\{ t=c \right\}$ with $t_0\leq c \leq \overline{t}$ is positive definite, then the term under the square root in \ref{ode-mu} is always positive. 
Moreover, since $f'_r>0$ then $\mu_r$ is future.

Let us show that we can find $r>0$ such that $\mu_r$ is twisted. 
A simple calculation gives
\[
(d\phi)_{\mu_r(s)} \left(\frac{D\mu'_{r}}{ds}\left(s\right)\right)=\left( f''_r+r^2\varphi_0\left(r,s\right),r\cos s + r^2\varphi_1\left(r,s\right), -r\sin s+r^2\varphi_2\left(r,s\right)  \right)
\]
where $\varphi_i=\varphi_i\left(r,s\right)$ with $i=0,1,2$ are continuous functions in $U$ depending on the Christoffel symbols and the components of $\mu'_r$. 
In order to show that $\frac{D\mu'_r}{ds}$ and $\mu'_r$ are linearly independent, it is enough to see that the determinant 
of their components $x$, $y$ does not cancel out, so
\begin{equation*}
\left|
\begin{array}{lr}
r\cos s + r^2\varphi_1\left(r,s\right) & r\sin s  \\
-r\sin s+r^2\varphi_2\left(r,s\right) & r \cos s 
\end{array}
\right| 
= r^2\left( 1  + r \left(\varphi_1\left(r,s\right)\cos s + \varphi_2\left(r,s\right)\sin s \right)\right)
\end{equation*}
hence, since $\varphi_1$ and $\varphi_2$ are continuous in $U$, they are also bounded in the compact set $U_0$ and there exists $r_0\leq R/2$ such that 
\[
1  + r \left(\varphi_1\left(r,s\right)\cos s + \varphi_2\left(r,s\right)\sin s \right)\neq 0
\] 
for all $r\in \left(0,r_0\right]$, and in this case, $\frac{D\mu'_{r}}{ds}$ and $\mu'_{r}$ are linearly independent.
 
At this moment, we have seen that $\mu_r$ is a twisted null curve passing through $\gamma\left(t_0\right)$ for $0<r\leq r_0$, and it remains to show that there exists $\delta>0$ such that $\mu_r$ also passes through $\gamma\left( t \right)$ for every $t\in \left(t_0,t_0+\delta \right]$.

Now, we want to prove that for every $r\in \left(0,r_0\right]$ there exists $s_r>0$ such that $f_r\left(s_r\right)=\overline{t}$.
Given $r\in \left(0,r_0\right]$, we define $\omega_r=\mathrm{sup}\left\{s:f_r\left(s\right) \hspace{1mm} \mathrm{exists}\right\}$. 
Let us assume that $\displaystyle{\lim_{s\mapsto \omega_r}f_r\left(s\right)}=c\leq \overline{t}$. 
In case of $\omega_r<+\infty$, the solution $\overline{f}_r$ of equation \ref{ode-mu} verifying the initial condition $\overline{f}_r\left(\omega_r\right)=c$ would coincide with $f_r=f_r\left(s\right)$ for $s<\omega_r$ contradicting the maximality of $f_r$ up to $\omega_r$ because in that case $f_r$ could be extended beyond $s=\omega_r$. 
On the other hand, if $\omega_r=+\infty$, the derivability of $f_r$ would imply that $\displaystyle{\lim_{s\mapsto +\infty}f'_r\left(s\right)}=0$ and hence the curve solution $\mu_r$ would approximate to the curve $\beta_r$ verifying 
\[
\beta_r\left(s\right)=\left(c,r\left(1-\cos s\right),r \sin s\right)\in U_0
\]
in $TM$, i.e. for every $s_0\in \mathbb{R}$ the sequence $\left\{ s_n = s_0+2\pi n \right\}_{n\in\mathbb{N}}$ would verify 
\[
\lim_{s\mapsto +\infty}\mu_r\left(s_n\right)=\beta_r\left(s_0\right) \hspace{5mm} \mathrm{and} \hspace{5mm} \lim_{s\mapsto +\infty}\mu'_r\left(s_n\right)=\beta'_r\left(s_0\right)
\] 
By the continuity of the metric $\mathbf{g}$ then we have 
\[
\lim_{s\mapsto +\infty}\mathbf{g}\left(\mu'_r\left(s_n\right),\mu'_r\left(s_n\right)\right)=\mathbf{g}\left(\beta'_r\left(s_0\right),\beta'_r\left(s_0\right)\right)\neq 0
\]
since $\beta_r$ is contained in the space-like hypersurface $\left\{t=c\right\}$, but this contradicts that $\mathbf{g}\left(\mu'_r,\mu'_r\right)=0$.
Therefore, independently from $\omega_r$, for every $r\in\left(0,r_0\right]$ we have that $\displaystyle{\lim_{s\mapsto \omega_r}f_r\left(s\right)}> \overline{t}$ and hence, for every $r\in\left(0,r_0\right]$ there exists $s_r\in \left(0,\omega_r\right)$ such that $f_r\left(s_r\right)=\overline{t}$.

Since the functions $g_{ij}$ are continuous in $U$ for $i,j=1,2$, then their restrictions to the compact set $U_0$ reach their maximum, therefore there exists $M_{ij}>0$ such that $\left|g_{ij}\left(t,x,y\right)\right|\leq M_{ij}$ for $\left(t,x,y\right)\in U_0$.
Then,
\begin{eqnarray*}
0<f'_r\left(s\right)&=&r\sqrt{ g_{11}\sin^2 s + 2 g_{12}\sin s \cos s +  g_{22}\cos^2 s }\leq \\
 &\leq& r\sqrt{ \left|g_{11}\sin^2 s \right| + 2 \left|g_{12}\sin s \cos s \right| +  \left|g_{22}\cos^2 s\right| } \leq \\
 &\leq& r\sqrt{ M_{11} + 2 M_{12} +  M_{22} } = rM 
\end{eqnarray*}
where $M=\sqrt{ M_{11} + 2 M_{12} +  M_{22} }\in \mathbb{R}$ is independent from $r$ and $s$.
So integrating, we have that $t_0\leq f_r\left(s\right) \leq rMs + t_0$ and therefore
\[
\overline{t}=f_r\left(s_r\right) \leq rMs_r + t_0 \hspace{5mm} \Rightarrow \hspace{5mm} \frac{\overline{t} - t_0}{rM}\leq s_r
\]
then there exists $\rho\in \left(0,r_0\right]$ small enough such that $s_r\geq 2\pi$ for all $r\in \left(0,\rho\right]$ and hence the parameter $s$ of $f_r$ can be extended beyond $s=2\pi$. 
Since $f'_{\rho}\left(s\right)>0$ then $f_{\rho}\left( s\right)>t_0$ for all $s>0$, therefore there exists $\delta>0$ such that $f_{\rho}\left(2\pi \right)=t_0+\delta$. 
So, by the inequality $t_0\leq f_r\left(2\pi\right) \leq 2\pi rM + t_0$ we have that $\lim_{r\mapsto 0}f_r\left(2\pi\right)=t_0$ and for every $t\in \left(t_0,t_0+\delta\right]$ there exists $r\in\left(0, \rho\right]$ such that 
\[
\mu_r\left(0\right)=\left(t_0,0,0\right)=\phi\left(\gamma\left(t_0\right)\right)
\]
\[
\mu_r\left(2\pi\right)=\left(f_r\left(2\pi\right),0,0\right)=\left(t,0,0\right)=\phi\left(\gamma\left(t\right)\right)
\]
therefore we have shown that there exists $\delta>0$ such that for every $t\in \left(t_0,t_0+\delta\right]$ the points $\gamma\left(t_0\right)$ and $\gamma\left(t\right)$ can be connected by some future twisted null curve $\mu_r$.
Analogously, this construction can be done to obtain a future twisted null curve joining $\gamma\left(t\right)$ to $\gamma\left(t_0\right)$ for all $t\in \left[t_0-\delta,t_0\right)$.
\end{proof}

\begin{lemma}\label{step2} The statement of Lemma \ref{step1} is true in a $m$--dimensional spacetime $M$.
\end{lemma}

\begin{proof}
We can find a synchronous coordinate system $\left(U,\phi\right)$ with $\phi = \left(t,x_1, \ldots, x_{m-1}\right)$ (as done previously) such that the expression of the geodesic $\gamma$ in these coordinates is $\phi\left(\gamma\left(s\right)\right)=\left(s,0,\ldots,0\right)\in \mathbb{R}^m$, so this chart is adapted to $\gamma$.
Consider the restriction 
\[
V=\left\{\left(t,x_1,\ldots,x_{m-1}\right):x_i=0, i=3,\ldots,m-1 \right\} \subset \phi\left(U\right)
\]
then $N=\phi^{-1}\left(V\right)\subset M$ is a $3$--dimensional manifold embedded in $M$. 
Moreover, by \cite[Lemma 4.3]{On83} we have that Levi-Civita connection in $N$ coincides with the orthogonal projection over $N$ of the Levi-Civita connection in $M$, hence we have $\frac{D^{N}}{ds}=\mathrm{tan}\left( \frac{D}{ds} \right)$ where $\frac{D^{N}}{ds}$ and $\frac{D}{ds}$ denote the covariant derivatives in $N$ and $M$ respectively. 
So the geodesics in $M$ contained in $N$ are also geodesics in $N$ and the restriction $\left(N,\left.\phi\right|_{N}=\left(t,x_1,x_2\right)\right)$ of the synchronous coordinate system is still a synchronous coordinate system for $N$. 
Then, since $\gamma$ is a geodesic contained in $N$, by step \ref{step1}, there exists $\delta>0$ and a future twisted null curve $\mu\subset N$ such that $\mu$ joins $\gamma\left( t_0 \right)$ to $\gamma\left( t_0 + \delta \right)$. 
Since the metric in $N$ is the restriction of the metric in $M$, then $\mu$ as curve in $M$ is also null. 
Finally, since $\mu'$ and $\frac{D^N \mu'}{ds}=\mathrm{tan}\left(\frac{D\mu'}{ds}\right)$ are lineally independent in $T_{\mu\left(s\right)}N$ then is an immediate consequence that $\mu'$ and $\frac{D\mu'}{ds}$ are lineally independent in $T_{\mu\left(s\right)}M$. Therefore, we have shown that there exists $\delta>0$ and $\mu$ a future twisted null curve in $M$ joining $\gamma\left( t_0 \right)$ to $\gamma\left( t_0 + \delta \right)$.
\end{proof}
 
We can prove now as a direct consequence of the previous lemmas, Lemma \ref{step1} and \ref{step2}, the following:

\begin{proposition}\label{mu-dim-3}
Let $\gamma:I\rightarrow M$ be a future timelike geodesic. Then, for any $t_0,t_1\in I$, there exists a future piecewise twisted null curve $\mu$ joining $\gamma\left( t_0 \right)$ to $\gamma\left( t_1 \right)$.
\end{proposition}

\begin{proof}
By Lemma \ref{step2}, for all $t\in \left[t_0,t_1\right]$ there exists an open interval $I_t=\left[t-\delta_t,t+\delta_t\right]\subset \left[t_0,t_1\right]$ relative to $\left[t_0,t_1\right]$ such that $\gamma\left(t\right)$ can be joined to $\gamma\left(u\right)$ with $u\in I_t$ by means of a piecewise twisted null curve. 
By the compactness of $\left[t_0,t_1\right]$, we can extract a finite covering $\left\{ I_n \right\}_{n=1,\ldots,N}$ such that, with no lack of generality, verifies $I_i \cap I_k \neq \varnothing \Leftrightarrow k=i\pm 1$.
We can choose a partition 
\[
\left\{ t_0=a_1 < b_1 < \cdots < a_{N-1} < b_{N-1} <a_N=t_1 \right\}
\]
such that $a_i \in I_i$ and $b_i\in I_i \cap I_{i+1}$ and therefore there exists future twisted null curves joining $\gamma\left(a_i\right)$ to $\gamma\left(b_i\right)$ and $\gamma\left(b_i\right)$ to $\gamma\left(a_{i+1}\right)$ for $i=1,\ldots,N-1$. The union of these curves forms a future piecewise twisted null curve connecting $\gamma\left(t_0\right)$ to $\gamma\left(t_1\right)$. 
\end{proof}

Now we can proceed with the proof of Theorem \ref{mu-teorema}.

\begin{proof}   Theorem \ref{mu-teorema}:
Consider $p,q\in M$ such that $q\in I^+(p)$, then there exists a continuous future time-like curve $\lambda$ connecting $p$ and $q$. 
By compactness of $\lambda$ between $p$ and $q$, there exists a finite covering $\left\{W_k\right\}_{k=1,\ldots, K}$ of globally hyperbolic and causally convex open sets, then it is possible to built a continuous curve $\gamma$ joining $p$ and $q$ 
formed by segments $\gamma_k \subset W_k$ of future time-like geodesics with endpoints at $\lambda$.  So $\gamma$ becomes a future piecewise time-like geodesic 

By Prop. \ref{mu-dim-3}, the endpoints of the time-like geodesic segments $\gamma_k$ of $\gamma$ can be connected by a future piecewise twisted null curve $\mu_k$. 
Since $\gamma$ is continuous, we can glue together all $\mu_k$ to obtain another piecewise twisted null curve $\mu$ joining $p$ and $q$.
\end{proof}

%%%%%%%%%%%%%%%%%%%%%%%%
%%%%%%%%%%%%%%%%%%%%%%%%

%\newpage

\section{The smooth structure of the space of skies and the non-refocussing property}\label{non-refocussing}

\subsection{Regular sets}
The smooth structure on the space of skies will be obtained by selecting a family of neighbourhoods possessing the properties that will make obvious the construction of an atlas on $\Sigma$.  We will call such neighbourhoods \emph{regular neighbourhoods} and they refine the notion of regular set already introduced in \cite[Def. 3]{Ba14}.

Let $W\subset \Sigma$ be a non-empty set satisfying the conditions:
\begin{enumerate}
\item \label{reg-1}$\widehat{T}X \cap \widehat{T}Y = \varnothing$ for all $X\neq Y\in W$.
\item \label{reg-2}The union 
\[
\widehat{W}=\bigcup_{X\in W} \widehat{T}X \subset \widehat{T}\mathcal{N}
\]
is a regular $\left(3m-4\right)$--dimensional submanifold of $\widehat{T}\mathcal{N}$.

\item \label{reg-3} Let $\widehat{\mathcal{D}}$ be the distribution in $\widehat{W}$ whose leaves are $\widetilde{X}=\widehat{T}X$.  Then the space of leaves $\widetilde{W}=\left\{ \widetilde{X}:X\in W \right\}=\widehat{W}/\widehat{\mathcal{D}}$ is a differentiable quotient manifold.
\end{enumerate}

It is clear that in this case, $\widetilde{W}$ can be identified to $W$ via the bijective map
\begin{equation}\label{ident-Sigma-tilde}
\begin{tabular}{rlcr}
$\Theta:$ & $W$ & $\rightarrow$ & $\widetilde{W}$ \\
 & $X$ & $\mapsto$ & $\widetilde{X}$
\end{tabular}
\end{equation}
and hence $W$ inherits the quotient topology such that
\begin{center}
$U\subset W $ is open $\displaystyle{ \Leftrightarrow \widehat{U} = \bigcup_{X\in U} \widehat{T}X \subset \widehat{W} }$ is open,
\end{center}
and also a differentiable structure from $\widetilde{W}$. 
So, we will denote $W$ equipped with the previous structure as $W^{\left( \sim \right)} \simeq \widetilde{W}$.
 
\begin{enumerate}
\setcounter{enumi}{3}
\item \label{reg-4} For every $X_0\in W$ and every celestial curve $\Gamma:I_{\epsilon}\rightarrow \mathcal{N}$ such that $\Gamma'\left(0\right)\in\widehat{T}X_0$,
\begin{enumerate}
\item \label{reg-4a}there exists $0<\delta \in I_{\epsilon}$ such that $\Gamma':I_{\delta}\rightarrow \widehat{W}$ with $I_{\delta}=\left(-\delta,\delta\right)$. 
\item \label{reg-4b}the curve $\chi_{X_0}^{\Gamma}:I_{\delta}\rightarrow W^{\left( \sim \right)}$ defined in Lemma \ref{mu-lemma} is differentiable.
\end{enumerate}
\item \label{reg-5} Given $\widetilde{X}, \widetilde{Y} \in \widetilde{W}$, for any causal curve $\chi\colon [a,b] \to \Sigma$, joining $X$ and $Y$,  then $\chi\left(s\right)\in W $ for all $s\in \left[a,b\right]$.
\end{enumerate}

Now we are ready to state the next definition:

\begin{definition}
A not--empty subset $W\subset \Sigma$ is said to be a \emph{regular} subset, and denoted as $W\subset_ {\mathrm{reg}} \Sigma$, if it verifies conditions (\ref{reg-1}) to (\ref{reg-5}) above. 
\end{definition}

Observe that both the definition of regular subset and the differentiable structure of $W^{\left( \sim \right)} \simeq \widetilde{W}$ depend only on $\mathcal{N}$ and $\Sigma$.

\subsection{The topology of the space of skies and regular sets}
We will show next that the class of regular subsets is not empty. 

We will say that $V\subset M$ is an open normal set is $V$ is globally hyperbolic, causally convex, relatively compact, open set of $M$.   A classical theorem due to Whitehead guarantees the existence of convex normal neighbourhoods $V$ at any point $x \in M$, (see \cite[chapter 5]{On83} and \cite[theorem 2.1 and definition 3.22]{Mi08} for a treatment of this result in Lorentz manifolds).  Thus for a strongly causal space-time $M$ there exists a basis of neighbourhoods at any $p\in M$ formed by normal open sets.

\begin{proposition}\label{prop-regular}
Let $V\subset M$ be a normal open set, then $U=S\left(V\right) \subset_\mathrm{reg} \Sigma$ is regular. Moreover, $S \colon V\rightarrow U^{\left(\sim\right)}$ is a diffeomorphism.
\end{proposition}

\begin{proof}
Let $V\subset M$ be a normal open set, then condition (\ref{reg-1}) is verified since $V$ is causally convex. 
By \cite[Thm. 1]{Ba14}, condition (\ref{reg-2}) is verified. 
The condition (\ref{reg-3}) and the fact of $S:V\rightarrow U^{\left(\sim\right)}$ being a diffeomorphism are consequences of \cite[Thm. 2]{Ba14}. 
Lemma \ref{mu-lemma} trivially implies (\ref{reg-4a}) and permits to construct the curve $\chi_{X_0}^{\Gamma}$ as the following composition of differentiable maps
\begin{equation*}
\begin{tabular}{ccccccc}
& $\Gamma$ & & $\pi$ & & $\Theta^{-1}$ & \\
$I_{\delta}$ & $\longrightarrow$ & $\widehat{U}$ & $\longrightarrow$ & $\widetilde{U}$ & $\longrightarrow$ & $U^{\left(\sim\right)}$  \\
$s$ & $\mapsto$ & $\Gamma'\left(s\right)$ & $\mapsto$ & $\widehat{T}\chi_{X_0}^{\Gamma}\left(s\right)$ & $\mapsto$ & $\chi_{X_0}^{\Gamma}\left(s\right)$
\end{tabular}
\end{equation*}
then (\ref{reg-4b}) is verified.
Finally, in order to verify (\ref{reg-5}), we know that $\Gamma'\left(a\right)\in \widehat{T}X$, $\Gamma'\left(b\right)\in \widehat{T}Y$ and $X,Y\in U$, by lemma \ref{mu-lemma}, there exists a piecewise twisted null curve $\mu:\left[a,b\right]\rightarrow M$ such that $\mu\left(a\right)=x\in V$ and $\mu\left(b\right)=y\in V$. 
Since $V$ is causally convex, then $\mu$ is fully contained in $V$ and therefore $\chi=S\circ \mu$ is fully contained in $U=S\left(V\right)$.
So, we conclude that $U \subset_\mathrm{reg} \Sigma$.
\end{proof}

We may call the regular sets $U = S(V)$ with $V$ open normal, elementary regular sets in $\Sigma$.

Using now the technical lemma:

\begin{lemma}\label{W-lemma}
Given $W \subset_\mathrm{reg} \Sigma$ a regular set and $X_0 = S\left(x_0\right)\in W$, then for any twisted null curve $\mu:I_{\epsilon}\rightarrow M$ such that $\mu\left(0\right)=x_0$ there exists $\delta>0$ verifying that $\mu\left(\left(-\delta,\delta\right)\right)\subset S^{-1}\left(W\right)$.
\end{lemma}

\begin{proof}
Consider $X_0 = S\left(x_0\right)\in W\subset_\mathrm{reg} \Sigma$, then by lemma \ref{mu-lemma}, there exists a celestial curve $\Gamma:I_{\epsilon}\rightarrow \mathcal{N}$ and a continuous curve $\chi_{X_0}^{\Gamma}:I_{\epsilon}\rightarrow \Sigma$ such that $\chi_{X_0}^{\Gamma}=S\circ \mu$.
Since $W$ is regular, then there exists $\delta>0$ such that $\chi_{X_0}^{\Gamma}:\left(-\delta,\delta\right)\subset I_{\epsilon}\rightarrow W^{\left(\sim\right)}$ is differentiable. 
Then we have 
\[
\mu\left(\left(-\delta,\delta\right)\right)=S^{-1}\circ \chi_{X_0}^{\Gamma} \left(\left(-\delta,\delta\right)\right)\subset S^{-1}\left(W^{\left(\sim\right)}\right)=S^{-1}\left(W\right).
\]
\end{proof}

\begin{figure}[b]
\centering
\includegraphics[width=9cm]{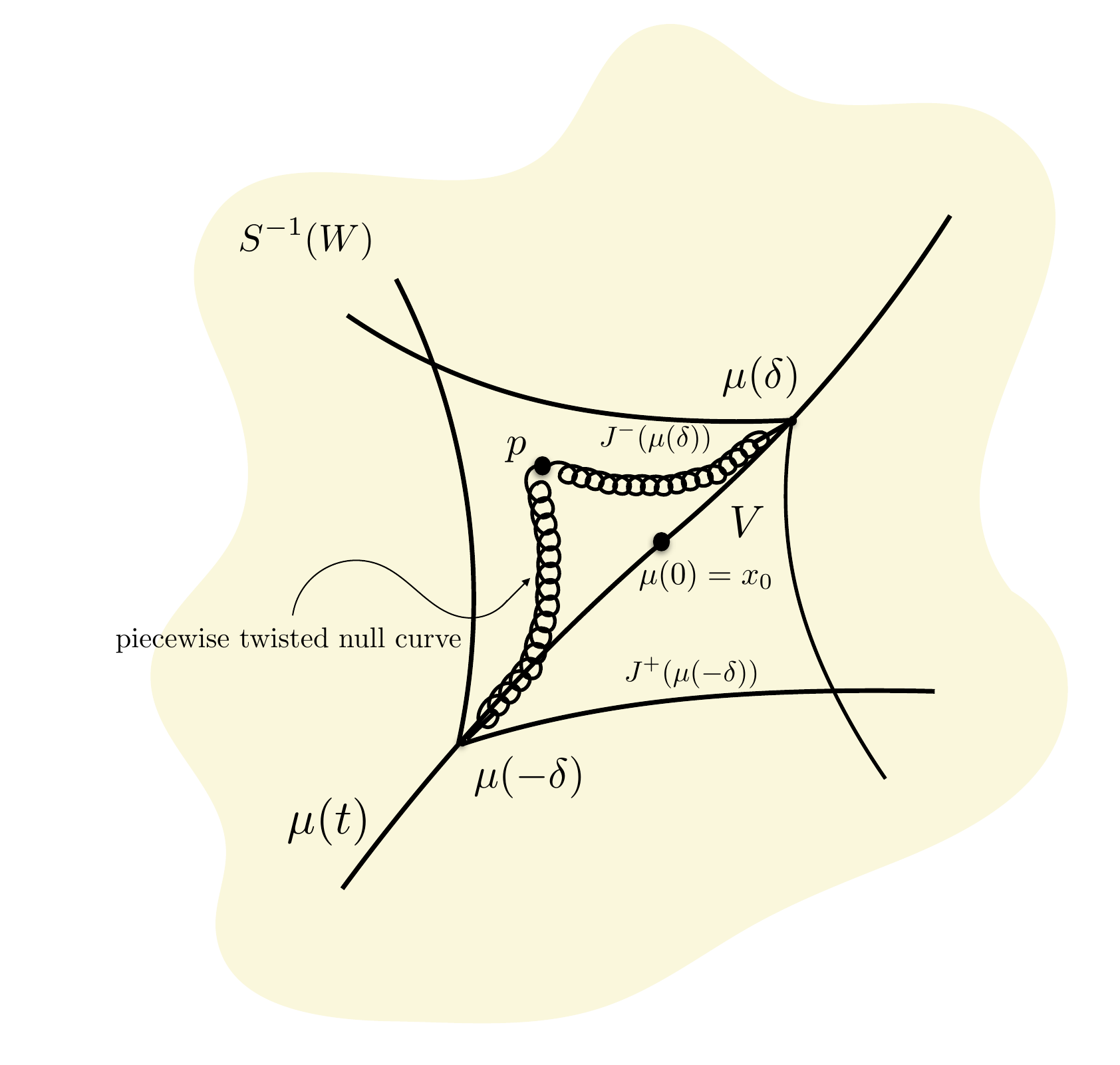}
\caption{}\label{fig:null_curves}
\end{figure}

It is easy to prove the following:
 
\begin{theorem}\label{theo-regular}
Let $W\subset_\mathrm{reg} \Sigma$ be a regular set, then $S^{-1}\left(W\right)$ is open in $M$.
\end{theorem}

\begin{proof}
Given $W\subset_\mathrm{reg} \Sigma$ and consider $X_0\in W$ such that $x_0=S^{-1}\left(X_0\right)\in M$.
Take a future twisted null curve $\mu:I_{\epsilon}\rightarrow M$ with $\mu\left(0\right)=x_0$ , then by lemma \ref{W-lemma}, there exists $\delta>0$ verifying that $\mu\left(\left(-\delta,\delta\right)\right)\subset S^{-1}\left(W\right)$. 
Without any lack of generality, we can assume that $\delta$ is small enough for $V=I^{+}\left(\mu\left(-\delta\right)\right)\cap I^{-}\left(\mu\left(\delta\right)\right)$ being globally hyperbolic and causally convex. 
Observe that $x_0\in V$ and for any $p\in V$, we have that $p\in I^+\left(\mu(-\delta)\right)$, then by theorem \ref{mu-teorema}, for any $p\in V$ there exists a future piecewise twisted null curve $\mu_p$ connecting $\mu\left(-\delta\right)$ and $\mu\left(\delta\right)$ passing through $p$  (see Figure \ref{fig:null_curves}).
Now, since $W$ is regular, then by property (\ref{reg-5}), the curve $\chi_p = S\circ \mu_p$ is fully contained in $W$, therefore $p\in S^{-1}\left(W\right)$ and hence $V\subset  S^{-1}\left(W\right)$ and $S^{-1}\left(W\right)$ is open in $M$.
\end{proof}

In virtue of Proposition \ref{prop-regular} and Theorem \ref{theo-regular}, since the sky map $S$ is an homeomorphism with the Low's topology in $\Sigma$, it is clear that this topology coincides with the topology generated by regular sets in $\Sigma$. 
So, by \cite[Cor. 1, Thm. 2 and Cor. 2]{Ba14}, we get the following corollary.

\begin{corollary}\label{differ}
The family of regular sets $\left\{ W \mid W \subset_\mathrm{reg}\Sigma \right\}$ is a basis for the Low's topology of $\Sigma$. 
Moreover, there exists a unique differentiable structure in $\Sigma$ compatible with the manifolds $W^{\left(\sim\right)}\subset \Sigma$ that makes of $S:M\rightarrow \Sigma$ a diffeomorphism.
\end{corollary}

In the previous construction of the topology of $\Sigma$ by mean of regular sets, the hypothesis of non--refocusing in $M$ has not been used, but we have obtained, with no further hypotheses, that the resultant topology coincides with Low's topology in $\Sigma$.

%So a consequence of this fact is the following result.
%\begin{corollary}
%If the skies of $M$ separate events then $M$ is non--refocusing.
%\end{corollary}

The following Lemma corroborates the relation between neighbourhood basis of $M$ and its space of skies $\Sigma$ and will be used to establish the conclusion that sky-separating implies non-refocussing.

\begin{lemma}\label{neigh_basis}
Let $\mathcal{B}\left(x\right)$ be a neighbourhood basis consisting on globally hyperbolic, normal and causally convex open sets.%\footnote{Creo que no es necesario imponer condiciones de hiperbolicidad global, convexidad causal y normalidad a los $U\in \mathcal{B}\left(x\right)$, pues \'{u}nicamente utilizamos que cada $\mathcal{U}\subset \mathcal{N}$ es abierto, y eso es cierto para cualquier abierto $U\subset M$.} of $x\in M$. 
For any $U\in \mathcal{B}\left(x\right)$, denote by $\mathcal{U}=\left\{ \gamma\in \mathcal{N}:\gamma\cap U\neq \varnothing \right\}$. Then $\left\{\Sigma\left( \mathcal{U} \right):U\in \mathcal{B}\left(x\right)\right\} $ is a neighbourhood basis of $S\left(x\right)\in \Sigma$. 
\end{lemma}

\begin{proof}  Because the bundle $\mathbb{PN}\left(M\right) \to M$ is locally trivial,
let us take a neighbourhood $V\subset M$ of $x\in M$ such that there is a diffeomorphism $\varphi:V\times \mathbb{S}^{m-2}\rightarrow\mathbb{PN}\left(V\right)$ with $\varphi\left(\left\{y\right\}\times \mathbb{S}^{m-2}\right)=\mathbb{PN}_y$ for all $y\in V$. 

Consider the map $\sigma:\mathbb{PN}\left(V\right)\rightarrow\mathcal{V}\subset\mathcal{N}$ defined by $\sigma\left(\left[v\right]\right)=\gamma_{\left[v\right]}$. 
It is clear that $\sigma$ is continuous and hence $\overline{\sigma}=\sigma\circ \varphi:V\times \mathbb{S}^{m-2}\rightarrow\mathcal{V}$ is also so.
Observe that  
\[
S\left(x\right)=\overline{\sigma} \left(\left\{x\right\}\times \mathbb{S}^{m-2}\right) \, ,
\]
and $\overline{\sigma}(V \times \mathbb{S}^{m-2}) = \mathcal{V}$.

Now, take any open $\mathcal{W} \subset \mathcal{V}$ containing the sky $S\left(x\right)$, then 
\[
\left\{x\right\}\times \mathbb{S}^{m-2} \subset\overline{\sigma}^{-1} \left(S\left(x\right)\right)\subset  \overline{\sigma}^{-1} \left(\mathcal{W}\right)
\]
Since $\overline{\sigma}$ is continuous then $\overline{\sigma}^{-1} \left(\mathcal{W}\right)$ is open in $V\times\mathbb{S}^{m-2}$. 

For any $\left(y,q\right)\in V\times\mathbb{S}^{m-2}$ there exists a neighbourhood basis whose elements are $U^{\left(y,q\right)}=K^y \times H^q$ where $K^y\subset V$ and $H^q\subset \mathbb{S}^{m-2}$ are open neighbourhoods of $y\in V$ and $q\in \mathbb{S}^{m-2}$ respectively. 
Then for any $(x,q) \in \{ x \} \times \mathbb{S}^{m-2}$, there exist $U^{(y,q)}$ with $(x,q) \in U^{(y,q)} \subset \overline{\sigma}^{-1}(W)$.   
Since $\left\{x\right\}\times\mathbb{S}^{m-2}$ is compact, then there exists a finite sub-covering $\left\{ U_j = K_j\times H_j\right\}_{j=1,\ldots,n}\subset \overline{\sigma}^{-1} \left(\mathcal{W}\right)$.
Then  
\[
\left\{x\right\}\times\mathbb{S}^{m-2} \subset \bigcup_{j=1}^{n}U_j \subset \overline{\sigma}^{-1} \left(\mathcal{W}\right)
\]
Observe that $K_0=\bigcap_{j=1}^{n}K_j$ is an open neighbourhood of $x$ and $\bigcup_{j=1}^{n}H_j = \mathbb{S}^{m-2}$.

Since $\mathcal{B}\left(x \right)$ is a neighbourhood basis of $x\in M$, there exists $U\in \mathcal{B}\left(x \right)$ such that $U\subset K_0$. 

For any $\left(y,q\right)\in U\times\mathbb{S}^{m-2}$, we have that 
\[
\left(y,q\right)\in  U\times \bigcup_{j=1}^{n}H_j
\]
therefore there exists $j$ such that $q\in H_j$ and since $y\in K_0\subset K_j$, then $\left(y,q\right)\in U_j \subset \overline{\sigma}^{-1}(W)$. 
This implies that 
\[
\left\{x\right\}\times\mathbb{S}^{m-2} \subset  U\times\mathbb{S}^{m-2} \subset \overline{\sigma}^{-1} \left(\mathcal{W}\right) .
\]
and hence
\[
S\left(x\right)\subset \overline{\sigma}\left( U\times \mathbb{S}^{m-2} \right) \subset  \mathcal{W}
\] 
and since $\mathcal{U} = \overline{\sigma}\left( U\times \mathbb{S}^{m-2} \right)$ then 
\[
S\left(x\right) \in \Sigma\left( \mathcal{U}  \right) \subset \Sigma\left( \mathcal{W} \right)
\]
is verified.  
Then $\left\{ \Sigma\left(\mathcal{U}\right):U\in \mathcal{B}\left(x\right) \right\} $ is a neighbourhood basis of $S\left(x\right)\in \Sigma$ as we claimed.
\end{proof}

A direct consequence of the previous results is the following:

\begin{theorem}\label{convergence}  Let $M$ be a space-time separating skies such that it is refocussing at $x$, then the sky map $S\colon M \to \Sigma$ is not open.
\end{theorem}

\begin{proof}  We will show that there exists a sequence $\{Êx_n\}$ in $M$ that does not converge to $x$ and such that $S(x_n)$ converges to $S(x)$ in $\Sigma$ does contradicting the statement that $S$ is open.

Because $M$ is refocussing at $x$ there exists an open neighbourhood $W \subset M$ of $x$ such that for every open neighbourhood $V \subset W$ of $x$ there is $y \notin W$ such that every light ray passing through $y$ enters $V$.  Let us choose a sequence of globally hyperbolic neighbourhoods $V_n^xÊ\subset W$ of $x$ such that $\cap_n V_n^x = \{Êx\}$.  More specifically, let $\sigma(t)$ be a time-like curve contained on a causally convex, globally hyperbolic neighbourhood $U\subset W$ of $x$ and let $a_n$ (respect. $b_n$) be a sequence of points on $\sigma$, in the past (future) of $x$, such that $a_n \to x$ (respect. $b_n \to x$).   Now we choose the sequence of open neighbourhoods as $V_n^x = I^+(a_n) \cap I^-(b_n)$. 

Then for any $V_n^x$ in the previous sequence there exists $x_n \notin W$ such that $\gamma \cap V_n^x \neq \emptyset$ and $x_n \in \gamma \in \mathcal{N}$.    Hence, since $x_n \notin W$ for all $n$, then $x_n$ cannot converge to $x$.  

On the other hand, considering the open subsets $\mathcal{U}_n = \{ \gamma \in \mathcal{N} \mid \gamma \cap V_n^x \neq \emptyset \}$, and because of Lemma \ref{neigh_basis}, it is clear that $\Sigma(\mathcal{U}_n)$ define a neighbourhood basis at $S(x)$ in $\Sigma$, and because $S(x_n) \in \Sigma(\mathcal{U}_n)$ then we conclude that $S(x_n) \to S(x)$.
\end{proof}

Then we get as a corollary of Thm. \ref{convergence}:

\begin{corollary}
If the skies of $M$ separate events then $M$ is non--refocusing.
\end{corollary}

\section{Conclusions and discussion}
We have reached the main conclusion that the topological, differentiable and causal structures of sky-separating strongly causal space-times can be reconstructed from the corresponding ones in their spaces of light rays and skies.     It is also important to point out that because of  Lemma \ref{neigh_basis} any strongly causal space-time is locally sky-separating, thus the property of being sky-separating has a global character.

The possibility of describing the causal structure of a space-time in terms of the partial order induced in the space of skies by non-negative Legendrian isotopies in the space of light rays, provides a new interpretation to the Malament-Hawking theorem, \cite{Ma77}, \cite{Ha76}, in the sense that the partial order relation defined on the space of skies  characterise the conformal structure of the original space-time.  Actually, suppose that $\Phi \colon \mathcal{N}_1 \to \mathcal{N}_2$ is a sky preserving diffeomorphism between the spaces of light rays of two strongly causal sky-separating space-times $M_1$ and $M_2$.   If the map $\Phi$ preserves the partial orders $\prec_a$, $a= 1,2$ defined in the spaces of skies $\Sigma_1$ and $\Sigma$, i.e., if $X\prec_1 Y$ then $\Phi(X) \prec_2 \Phi(Y)$, for any $X,Y \in \Sigma_1$, then because of Cor. \ref{order_iso}, we have that $\Phi$ induces a causal diffeomorphism $\varphi\colon M_1 \to M_2$, hence a conformal diffeomorphism.

The characterisation of causal relations in terms of sky isotopies opens a new direction in the foundations of the causal sets programme to quantum gravity \cite{Br91}, \cite{Ri00}, as it shows that causal structures need for their description the additional structure provided by the contact structure in the space of light rays.

It is also worth to point out here that the causal completion of a given spacetime is just continuous and often fails to be smooth (as in the case of Minkowski) space.   According to the reconstruction theorems discussed in this  paper a similar analysis could be performed directly on the space of light rays and skies. In this setting a concrete proposal of a new causal boundary construction was proposed by R. Low \cite{Lo06} but has not been discussed in detail so far.

A particularly interesting situation happens for three dimensional space-times that will be discussed in a forthcoming paper.  In such case the space of light rays happens to be three dimensional again as well as the space of skies.    In such case Low's causal boundary can be constructed explicitly and their topology can then be compared with that of the original space-time.

\end{document}